\documentclass[runningheads]{llncs}
\usepackage[T1]{fontenc}
\usepackage{graphicx}

\usepackage{wrapfig}
\usepackage{enumitem}

\usepackage[firstpage]{draftwatermark}
\SetWatermarkText{\raisebox{217mm}{%
		\hspace{95mm}%
		\href{https://eapls.org/pages/artifact_badges/}{\includegraphics[width=0.10\textwidth]{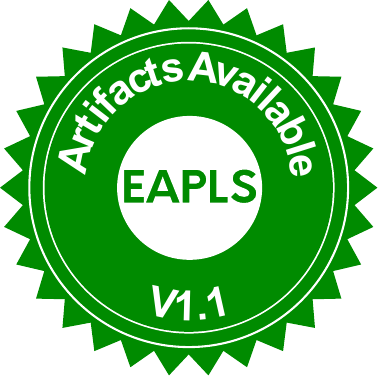}}%
		\href{https://eapls.org/pages/artifact_badges/}{\includegraphics[width=0.10\textwidth]{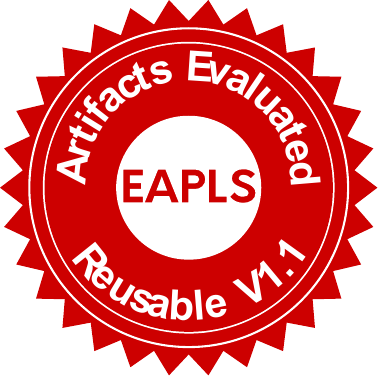}}%
}}
\SetWatermarkAngle{0}

\usepackage{tikz}
\usetikzlibrary{calc, patterns, arrows, shapes, positioning, fit, cd, tikzmark, decorations.pathreplacing, calligraphy, matrix, decorations.pathmorphing, shapes.multipart, plotmarks}
\usepackage{pgfplots}
\usepackage{hyperref}
\usepackage{color}

\usepackage[author=anonymous,marginclue,footnote]{fixme}
\FXRegisterAuthor{mh}{amh}{MH}
\FXRegisterAuthor{sp}{asp}{SP}
\FXRegisterAuthor{dh}{adh}{DH}
\FXRegisterAuthor{ls}{als}{LS}

\usepackage{stmaryrd}
\usepackage{amsmath}
\usepackage{amssymb}

\usepackage[frozencache,cachedir=.]{minted}

\usepackage{booktabs, tabularx}

\usepackage{listings}
\usepackage{color}
\definecolor{lightgray}{rgb}{0.9,0.9,0.9}
\spnewtheorem{thm}[theorem]{Theorem}{\bfseries}{\itshape}
\spnewtheorem{cor}[theorem]{Corollary}{\bfseries}{\itshape}
\spnewtheorem{cnj}[theorem]{Conjecture}{\bfseries}{\itshape}
\spnewtheorem{lem}[theorem]{Lemma}{\bfseries}{\itshape}
\spnewtheorem{lemdefn}[theorem]{Lemma and Definition}{\bfseries}{\itshape}
\spnewtheorem{prop}[theorem]{Proposition}{\bfseries}{\itshape}
\spnewtheorem{defn}[theorem]{Definition}{\bfseries}{\upshape}
\spnewtheorem{rem}[theorem]{Remark}{\bfseries}{\upshape}
\spnewtheorem{notation}[theorem]{Notation}{\bfseries}{\upshape}
\spnewtheorem{expl}[theorem]{Example}{\bfseries}{\upshape}
\spnewtheorem{thmdefn}[theorem]{Theorem and Definition}{\bfseries}{\itshape}
\spnewtheorem{propdefn}[theorem]{Proposition and Definition}{\bfseries}{\itshape}
\spnewtheorem{assn}[theorem]{Assumption}{\bfseries}{\upshape}
\spnewtheorem{algorithm}[theorem]{Algorithm}{\bfseries}{\upshape}

 \renewenvironment{corollary}{\begin{cor}}{\end{cor}}
 \renewenvironment{lemma}{\begin{lem}}{\end{lem}}

\newcommand{\sem}[1]{\llbracket #1 \rrbracket}
\newcommand{\Agents}[0]{\ensuremath{\mathsf{Ag}}}
\newcommand{\agent}[0]{\ensuremath{a}}
\newcommand{\Atoms}[0]{\ensuremath{\mathsf{At}}}
\newcommand{\Vars}[0]{\ensuremath{\mathsf{V}}}
\newcommand{\CLbox}[1]{\ensuremath{[#1]}}
\newcommand{\CLdia}[1]{\ensuremath{\langle#1\rangle}}
\newcommand{\Pow}[0]{\ensuremath{\mathcal{P}}}
\newcommand{\G}[0]{\ensuremath{\mathcal{G}}}

\newcommand{\gamecgf}{\G^{\mathsf{CGF}}}
\newcommand{\gameef}{\G^{\mathsf{EF}}}

\newcommand{\mycomment}[1]{}

\usepackage[capitalise,nameinlink]{cleveref}
\crefname{defn}{Definition}{Definition}

\begin{document}
\title{Efficient
  Model Checking\\ for the Alternating-Time $\mu$-Calculus\\ via
  Effectivity Frames}
\titlerunning{Efficient Model Checking for the Alternating-Time $\mu$-Calculus}
%
\author{Daniel Hausmann$^1$\thanks{Supported
		by the ERC Consolidator grant
		D-SynMA (No. 772459) and by the EPSRC through grant EP/Z003121/1} \and
 Merlin Humml$^2$\thanks{Funded by the German Federal Chamber of Notaries, project DIREGA} \and
 Simon Prucker$^2$\thanks{Funded by the Deutsche Forschungsgemeinschaft (DFG, German
   Research Foundation) – project number 517924115}
 \and
 Lutz Schröder$^2$\thanks{Funded by the Deutsche Forschungsgemeinschaft (DFG, German
   Research Foundation) – project number 419850228}
}
 \authorrunning{D.\ Hausmann, M.\ Humml, S.\ Prucker, L.\ Schröder}
%
\institute{$^1$ University of Liverpool, United Kingdom\\
$^2$ Friedrich-Alexander-Universität Erlangen-Nürnberg, Germany}

\maketitle              
\vspace{-10pt}
\begin{abstract}
  The semantics of alternating-time temporal logic (ATL) and the more
  expressive alternating-time $\mu$-calculus (AMC) is standardly given
  in terms of concurrent game frames (CGF). The information required
  to interpret AMC formulas is equivalently represented in terms of
  effectivity frames in the sense of Pauly; in many cases, this
  representation is more compact than the corresponding CGF, and in
  principle allows for faster evaluation of coalitional modalities. In
  the present work, we investigate whether implementing a model
  checker based on effectivity frames leads to better performance in
  practice.  We implement the translation from concurrent game frames
  to effectivity frames and analyse performance gains in model
  checking based on corresponding instantiations of a generic model
  checker for coalgebraic $\mu$-calculi, using dedicated benchmark
  series as well as random systems and formulas. In the process, we
  also compare performance to the state-of-the-art ATL model checker
  MCMAS. Our results indicate that on large systems, the overhead
  involved in converting a CGF to an effectivity frame is often
  outweighed by the benefits in subsequent model checking.

\end{abstract}
\section{Introduction}\label{sec:intro}

\noindent Alternating-time temporal logic (ATL) and its extension with
full fixpoints, the alternating-time $\mu$-calculus (AMC) play an
established role as a core formalism for the specification of
multi-agent systems~\cite{AlurEA02}. They allow expressing the
ability of groups of agents (\emph{coalitions}) to achieve short-term
and long-term goals by coordinated action. The semantics of ATL and
the AMC is standardly defined in terms of \emph{concurrent game
  frames} (CGFs), which at each state assign to each agent an explicit
set of available moves, along with an outcome function
determining which successor state is reached once all agents have
picked a move. It has been shown that for purposes of evaluating AMC
formulas, one may equivalently convert a CGF into an \emph{effectivity
  frame}, which at each state records, for each coalition~$C$, a set
of sets of states, understood as consisting of those state properties
that~$C$ may (alternatively) enforce by joint action in the next
state~\cite{Pauly02}.

Like in the case of standard temporal logics, one of the key
verification tasks in this context is \emph{model checking}, i.e.\ to
determine whether a given state in a multi-agent system, say a CGF,
satisfies a given formula of ATL or the AMC. The most natural way to
present the input system for purposes of ATL or AMC model checking is,
presumably, indeed to specify a CGF, which makes the actions available
to the agents and their effects explicit. On the other hand, it is
immediate from inspection of the respective semantics that evaluation
of alternating-time modalities, which over CGFs amounts to the
evaluation of nested quantifiers over actions (``coalition~$C$ has a
joint action such that no matter what the other agents do, the target
formula holds'') and hence takes exponential time in the number of
agents, is computationally more efficient on effectivity frames, where
it reduces to a simple look-up operation.

In the present work, we explore whether this observation can be
translated into actual efficiency gains in model checking. That is, we
trade the computational cost of evaluating coalitional modalities on
CGFs for an additional preprocessing step in which we translate a CGF
into an effectivity frame. The computations involved in the
translation are, prima facie, similar to the ones involved in
evaluating coalitional modalities on CGFs, but they are performed only
once, while especially in the AMC, modalities may be evaluated
repeatedly at the same state. We provide an implementation of AMC
model checking both on CGFs and on effectivity frames, building on a
generic implementation of model checking algorithms for coalgebraic
$\mu$-calculi~\cite{DBLP:conf/concur/HausmannS19} within the
\emph{Coalgebraic Ontology Logic Solver
  (COOL)}~\cite{DBLP:conf/cade/GorinPSWW14,COOL2,COOLMC}. Thereby, we incidentally
provide, as far as we are aware, the first model checker that covers
the entire AMC (which contains ATL$^*$ as a
fragment~\cite[Theorem~6.1]{AlurEA02}), rather than only ATL; furthermore,
our tool is agnostic as to whether the input multi-agent system
is provided in the form of a CGF or as an effectivity frame and hence enables  working directly with effectivity frames. We run
extensive experiments both on random models and formulas and on
dedicated benchmark series; besides comparing between CGF-based and
effectivity-frame-based model checking, we also compare the
performance of our implementation with the state-of-the-art ATL model
checker MCMAS~\cite{DBLP:journals/sttt/LomuscioQR17}. As expected,
results vary strongly with the exact nature of benchmark series, in
particular in the comparison with MCMAS; notably, MCMAS implements
symbolic model checking while COOL offers only explicit-state model
checking, and correspondingly MCMAS performs better on problems that
can be succinctly represented symbolically while COOL plays out
advantages when this is not the case.

\paragraph*{Related Work} As indicated above, the notion of
effectivity frame semantics was introduced by
Pauly~\cite{Pauly02}. Goranko and Jamroga~\cite{Goranko_2004} provide
a comparison of ATL semantics over CGFs, effectivity frames, and
\emph{alternating transition systems} (ATSs), in which the moves of
the agents are subsets of the state space. The translation from CGFs
to ATSs involves a blowup of the state space, so we refrain from
including ATS semantics in the present comparison. As far as we are
aware, the first implemented model checker for ATL was
MOCHA~\cite{DBLP:conf/icse/AlurAGHKKMMW01}, which however seems to be
no longer available. Kański et al.~\cite{unboundedatl2021} compare
MCMAS to an ATL model checker UMC4ATL based on McMillan's method of
\emph{unbounded model checking}~\cite{DBLP:conf/cav/McMillan02}. The
results indicate better performance of MCMAS, on which we therefore
focus in our comparative evaluation. The complexity of model checking
ATL and related logics has been investigated extensively already at
the time of introduction~\cite{AlurEA02}; we additionally mention work
on symbolic model checking for~ATL~\cite{DBLP:conf/atal/HoekLW06} and
on the complexity of explicit-state model checking for
ATL$^+$~\cite{DBLP:journals/iandc/GorankoKR21}. 
The AMC (hence also ATL*) is subsumed by the more expressive
strategy logic (SL)~\cite{CHATTERJEE2010677}, which allows the specification
of strategies that fail to be $\omega$-regular, 
but has nonelementary model checking~\cite{DBLP:journals/tocl/MogaveroMPV14};
restricted fragments of SL with more favourable model checking complexities
have been considered~\cite{DBLP:conf/cav/CermakLMM14,DBLP:conf/ijcai/BelardinelliJKM19}
 (and implemented in MCMAS-SLK~\cite{DBLP:conf/cav/CermakLMM14}). 

\section{The Alternating-Time $\mu$-Calculus}\label{sec:amc}

\noindent We briefly recall the syntax and semantics of the
\emph{alternating-time $\mu$-calculus}, a fixpoint logic
that allows the specification of $\omega$-regular joint strategies of coalitions
in multi-agent systems~\cite{AlurEA02}. Its key
feature are \emph{coalitional modalities} $\CLbox{C}$ `coalition~$C$
of agents is able to enforce'; temporal idioms over this base are
expressed using least and greatest fixpoint operators. The alternating-time $\mu$-calculus
contains \emph{alternating-time temporal logic (ATL)} as a fragment,
embedded in essentially the same way as computational tree logic
embeds into the standard $\mu$-calculus.

\paragraph{Syntax}
Formulas of the \emph{alternating-time temporal $\mu$-calculus~(AMC)} are given by the following grammar, depending on (countable) sets
\Agents{}, \Atoms{} and \Vars{} of \emph{agents}, \emph{propositional atoms} and \emph{fixpoint variables}, respectively.
\begin{gather*}
    \varphi, \psi ::= \top \mid \bot \mid p \mid \neg p\mid \varphi \land \psi 
    \mid \varphi \lor \psi \mid \CLbox{C}\varphi \mid \CLdia{C}\varphi\mid X\mid
    \eta X.\,\varphi 
\end{gather*}
where $p \in \Atoms$, $X\in\Vars$, $\eta\in \{\mu,\nu\}$, and~$C$ ranges over
\emph{coalitions}, i.e.\  subsets of $\Agents$.
\emph{Coalitional modalities} $\CLbox{C}\varphi$ (for which we employ
notation as in \emph{coalition logic}~\cite{Pauly02}) intuitively
express that the agents in coalition $C$ have a joint strategy to
enforce $\varphi$ in the next step of the game; dually,
$\CLdia{C}\varphi$ expresses that coalition~$C$ cannot prevent that
$\varphi$ is satisfied in the next step.  The fixpoint operators give
rise to standard notions of \emph{bound} and \emph{free} (occurrences of)
fixpoint variables; a formula is \emph{closed} if it does not contain
 free fixpoint variables, and \emph{clean} if every
fixpoint variable in it is bound by at most one fixpoint operator.
Given a closed formula $\varphi$, we let $|\varphi|$ denote its
syntactic size. The algorithms and implementations in this paper work
on the \emph{closure} $\mathsf{cl}(\varphi)$ of $\varphi$. The closure
is a succinct graph representation of the respective formula,
intuitively obtained from its syntax tree by identifying occurrences
of fixpoint variables with their binding fixpoint operators; we have
$|\mathsf{cl}(\varphi)|\leq |\varphi|$.  Finally, we define the
\emph{alternation-depth} $\mathsf{ad}(\varphi)$ of fixpoint formulas
$\varphi=\eta X.\,\psi$ in the usual way, capturing the number of
alternations between least and greatest fixpoints; for a detailed
account of these syntactic notions, see~\cite{KupkeMV22}.

As mentioned in \cref{sec:intro}, the AMC has various equivalent
semantics; we recall the semantics based on \emph{concurrent game
  frames}~\cite{AlurEA02} and that based on \emph{effectivity
  frames}~\cite{Pauly02}.  Let $W$ be a non-empty finite set of states. A
model $\mathcal{M} = (\mathcal{F}, \rho)$ consists of a \emph{valuation}
$\rho: \Atoms \to \mathcal{P}(W)$ that assigns to each propositional atom
$a\in\Atoms$ the set $\rho(a)$ of states that satisfy $a$; $\mathcal{F}$
is either a concurrent game frame (over $W$) or an effectivity frame
(over $W$), introduced next.

\paragraph{Concurrent game frame semantics:}
For a coalition~$C$, we write $\Pi_C=\mathbb{N}^{|C|}$ to denote the
set of \emph{joint (one-step) strategies} (or \emph{joint moves})
of~$C$; we denote the set $\Pi_\Agents$
of \emph{complete (one-step) strategies} (or \emph{grand moves}) by just $\Pi$. Given a joint strategy
$s_C\in \Pi_C$ of~$C$, any joint strategy
$s_{\overline{C}}\in \Pi_{\overline{C}}$ of the
counter coalition $\overline{C}=\Agents\setminus C$ induces a
complete strategy $(s_C,s_{\overline{C}})\in\Pi$. In this notation, a
\emph{concurrent game frame (CGF)}
    \begin{align*}
    \mathcal{F} = (W,m:W\times\Agents\to\mathbb{N}, f:W\times\Pi\rightharpoonup W)
    \end{align*}
    consists of $W$ together with functions
    \begin{itemize}
    \item $m:W\times\Agents\to\mathbb{N}$, assigning to each state
      $w\in W$ and each agent $a\in\Agents$ the number $m(w,a)>0$ of
      moves available to agent $a$ at the state $w$; and
    \item $f:W\times\Pi\rightharpoonup W$, assigning to each state
      $w\in W$ and each complete strategy
      $s=(s_1,\ldots,s_{|\Agents|})\in\Pi$ that is \emph{admissible} at $w$,
      i.e.\ $s_i< m(w,i)$ for all $i\in \Agents$, an \emph{outcome} state
      $f(w,s)\in W$.
    \end{itemize}
    Satisfaction $w,\sigma\models\varphi$ of an AMC formula
    $\varphi$ in a state~$w$ of a CGF $\mathcal{F}=(W,m,f)$ is defined
    inductively, depending on a \emph{valuation}
    $\sigma\colon\Vars\to\Pow(W)$ of the fixpoint variables; we denote
    the \emph{extension} of $\varphi$ under $\sigma$ by
    $\sem{\varphi}_\sigma=\{w\in W\mid w,\sigma\models \varphi\}$. For
    a coalition~$C$ and a state $w\in W$, we let $\Pi^w_C$ denote the
    set of joint strategies of~$C$ that are admissible at $w$.  The
    semantic clauses for propositional atoms and connectives are as usual
    (the former is given by the valuation $\rho$ of atoms);
    modalities are interpreted as
    \begin{align*}
     \sem{\CLbox{C}\varphi}_\sigma&=\{w\in W\mid \exists s_C\in \Pi^w_C.\,\forall s_{\overline{C}}\in
     \Pi^w_{\overline{C}}.\,f(w,(s_C,s_{{\overline{C}}}))\in\sem{\varphi}_\sigma\}\\
     \sem{\CLdia{C}\varphi}_\sigma&=\{w\in W\mid \forall s_C\in \Pi^w_C.\,\exists s_{\overline{C}}\in
     \Pi^w_{\overline{C}}.\,f(w,(s_C,s_{{\overline{C}}}))\in\sem{\varphi}_\sigma\}
    \end{align*}
    (where we write $(s_C,s_{\overline C})$ for the complete strategy extending~$s_C$ and~$s_{\overline C}$) so that, e.g., $w,\sigma\models\CLbox{C}\varphi$ if and only if coalition~$C$ has a joint
    strategy~$s_C$ at~$w$ such that for all strategies that complete~$s_C$ and are admissible
    at~$w$, the outcome satisfies~$\varphi$ under the valuation~$\sigma$ of fixpoint variables;
    and the semantics of fixpoint operators is given (exploiting the Knaster-Tarski fixpoint theorem) by
 	\begin{align*}
     \sem{\mu X.\,\varphi}_\sigma  & =
     \textstyle\bigcap \{Z\subseteq W\mid \sem{\varphi}^X_\sigma (Z)\subseteq Z\} &
     \sem{\nu X.\,\varphi}_\sigma &=
     \textstyle\bigcup \{Z\subseteq W\mid Z\subseteq\sem{\varphi}^X_\sigma (Z)\},
    \end{align*}
    where $\sem{\varphi}^X_\sigma$ is the function defined by
    $\sem{\varphi}^X_\sigma(A)= \sem{\varphi}_{\sigma[X\mapsto A]}$
    for $A\subseteq W$, where $\sigma[X\mapsto A](X)=A$ and
    $\sigma[X\mapsto A](Y)=\sigma(Y)$ for $Y\neq X$.  For a
    \emph{closed} formula~$\varphi$ (in which all fixpoint variables
    $X$ occur in the scope of some binding fixpoint operator $\mu X$
    or $\nu X$), $\sem{\varphi}_\sigma$ does not depend on $\sigma$,
    so in this case we just write $\sem{\varphi}$ (or $\sem{\varphi}_{\mathcal{F}}$
    to make the semantic domain explicit) for $\sem{\varphi}_\sigma$.
	
\paragraph{Effectivity frame semantics:}
    An \emph{effectivity frame (EF)} 
    \begin{align*}
    \mathcal{F} = (W, e\colon W \times \mathcal{P}(\Agents) \to \mathcal{P}(\mathcal{P}(W)))
    \end{align*}    
    consists of $W$ together with an effectivity function $e$ that
    assigns, to each state $w\in W$ and each coalition
    $C\subseteq\Agents$, a set $e(w,C)$ of sets of states.
    Intuitively, we will have $U\in e(w,C)$ whenever coalition $C$ has
    a joint strategy $s_C$ at $w$ that guarantees that the next state
    in the game is contained in $U$. To support this intuitive
    meaning, effectivity frames need to be restricted to be
    \emph{playable}, ensuring for finite~$W$ that they are induced by
    a CGF~\cite{Pauly02,Goranko_2004,DBLP:journals/lmcs/LitakPSS18}.
    When constructing effectivity frames from CGFs, playability is automatic.
           
    Satisfaction of AMC formulas over an effectivity frame
    $\mathcal{F}=(W,e)$ is defined inductively by the same clauses for
    propositional and fixpoint operators as for CGFs, but 
    modalities are interpreted as
    \begin{align*}
	 \sem{\CLbox{C}\varphi}_\sigma &= \{ w\in W\mid 
     \exists U\in e(w,C).\, U\subseteq\sem{\varphi}_\sigma\}\\
	 \sem{\CLdia{C}\varphi}_\sigma &= \{ w\in W\mid 
     \forall U\in e(w,C).\, U\cap\sem{\varphi}_\sigma\neq \emptyset\}
    \end{align*}
 so that with this semantics, $w,\sigma\models\CLbox{C}\varphi$ if and only if
 coalition $C$ can enforce some set $U$ at state $w$ such that all states contained
 in $U$ satisfy $\varphi$ (under the valuation $\sigma$ of fixpoint variables).
 The difference between the two semantics is displayed in \cref{fig:sem}, which illustrates the encoding of a single game step in both representations.
    \begin{figure}[ht]
      \tikzset{world/.style={draw, circle},nei/.style={draw, rounded corners, inner sep = 1em}}
        \centering
        \begin{scriptsize}
        \begin{tikzpicture}
          \node[world] (w1) at (0,1) {\(w_1\)};
          \node[world, label={right: \(p\)}] (w2) at (-.85,2) {\(w_2\)};
          \node[world, label={left: \(q\)}] (w3) at (.85,2) {\(w_3\)};

          \draw[-Latex, bend left] (w1) to node[below left, align=center] {\((1,2,1)\),\((2,2,1)\),\\\((2,1,1)\),\((1,1,1)\),\\\((1,1,2)\)} (w2);
          \draw[-Latex, bend right] (w1) to node[below right, align=center] {\((1,2,2)\),\((2,2,2)\),\\\((2,1,2)\)} (w3);

			\node[world] (w1a) at (6,0) {\(w_1\)};
			\node[world, label={right: \(p\)}] (w2a) at (-.85+6,2.5) {\(w_2\)};
			\node[world, label={left: \(q\)}] (w3a) at (.85+6,2.5) {\(w_3\)};
		
			\node [nei, fit=(w2a)] (n2a) {};
			\node [nei, fit=(w3a)] (n3a) {};
			\node [nei, fit=(n2a)(n3a)] (n23a) {};
		
			\draw[-Latex, bend left] (w1a) to node [right,align=center] {\(\{1\}\),\(\{2\}\),\\\(\{3\}\)} (n23a);
			\draw[-Latex, bend left=40] (w1a) to node [below left, align=center] {\(\{1,2\}\),\(\{3\}\)} (n2a);
			\draw[-Latex, bend right=40] (w1a) to node [below right, align=center] {\(\{2,3\}\),\(\{1,3\}\)} (n3a);
	\end{tikzpicture}
\end{scriptsize}

        \caption{Concurrent game frame semantics vs. effectivity frame semantics} \label{fig:sem}
\end{figure}
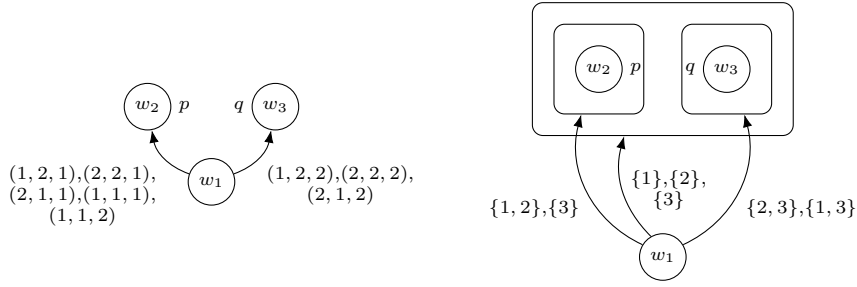
 In the CGF encoding on the left hand side, we have, e.g., $w_1\models\CLbox{\{1,3\}} q$ since, for the joint strategy that agents~$1$ and~$3$
both pick their second move, the outcome  is always~$w_3$, no matter which move agent~$2$ picks.
This is reflected in the EF encoding by $\{1,3\}$ being able to force the singleton set $\{w_3\}$.

Concurrent game frames can be  transformed into equivalent effectivity frames~\cite{Pauly02};
explicitly, the conversion works as follows.
 
\begin{defn}[Induced effectivity frame]\upshape\label{def:translation}
  A CGF $\mathcal{F}=(W,m,f)$ \emph{induces} the effectivity frame
  $\mathcal{F}'=(W,e)$ defined by
  \begin{align*}
e(w,C)=\{\{f(w,s_C,s_{\overline{C}})\mid s_{\overline{C}}\in \Pi^w_{\overline{C}}\}\mid s_C\in \Pi^w_C\}
\end{align*}
for $w\in W$ and $C\subseteq\Agents$.
\end{defn}

\begin{lemma}\label{lem:correctness}
  Let $\mathcal{F}$ be a concurrent game frame and let $\mathcal{F}'$
  be the effectivity frame induced by $\mathcal{F}$, and let~$\sigma$
  be a valuation. Then for all formulas $\varphi$,
  \begin{equation*}
    \sem{\varphi}_{\mathcal{F},\sigma}=\sem{\varphi}_{\mathcal{F'},\sigma}.
  \end{equation*}
\end{lemma}
\begin{proof}
The claim follows immediately from the equi-satisfaction of modalities in $\mathcal{F}$
and $\mathcal{F}'$. For readability,
we show this for closed modal formulas $\CLbox{C}\psi$. 
Assuming that $\sem{\psi}_{\mathcal{F}}=\sem{\psi}_{\mathcal{F}'}$, we have
\begin{align*}
	\sem{\CLbox{C}\psi}_{\mathcal{F}} & =
	\{w\in W\mid \exists s_C\in \Pi_C^w.\,
	\forall s_{\overline{C}}\in \Pi_{\overline{C}}^w.\,
	f(w,(s_C,s_{\overline{C}}))\in\sem{\psi}_{\mathcal{F}}\}\\
	&=
	\{w\in W\mid \exists s_C\in \Pi_C^w.\,\forall u\in \{f(w,s_C,s_{\overline{C}})\mid s_{\overline{C}}\in \Pi_{\overline{C}}^w\}.\,u\in\sem{\psi}_{\mathcal{F}'}\}\\
	&=
	\{w\in W\mid \exists U\in \{\{f(w,s_C,s_{\overline{C}})\mid s_{\overline{C}}\in \Pi^w_{\overline{C}}\}\mid s_C\in \Pi^w_C\}.\,U\subseteq\sem{\psi}_{\mathcal{F'}}\}\\
	&=
	\{w\in W\mid \exists U\in e(w,C).\,
	U\subseteq\sem{\psi}_{\mathcal{F}'}\}\\
	&=\sem{\CLbox{C}\psi}_{\mathcal{F}'},
\end{align*}
where the second equivalence holds
by assumption.\qed
\end{proof}

\section{Model Checking for the AMC}\label{sec:mc}

The model checking problem for the AMC consists in deciding, for state
$w$ in a model given either as a CGF or as an effectivity frame, and a
closed AMC formula $\varphi$, whether $w\in\sem{\varphi}$. The problem
is known to be in $\textsc{NP}\cap\textsc{co-NP}$ for the AMC, and in
$\textsc{PTime}$ for ATL.  A standard way to decide the problem is by
reduction to parity games, enabling the use of recent
quasipolynomial game-solving algorithms for model checking in the AMC.

\paragraph{Parity Games:}
A \emph{parity game} \(\mathcal{G} = (V,V_{\exists}, E,\Omega)\) is an
infinite-duration two-player game, played by the players~$\exists$ and~$\forall$
(\emph{Eloise} and \emph{Abelard}). It consists of a set~$V$ of
positions, with positions $V_{\exists}\subseteq V$ owned by~$\exists$
and the others by~$\forall$, a \emph{move} relation
$E\subseteq V\times V$, and a priority function
\(\Omega\colon V\to \mathbb{N}\).  A \emph{play} is a path in the
directed graph $(V,E)$ that is either infinite or ends in a node
$v\in V$ with no outgoing moves. Finite plays \(v_0v_1\dots v_n\) are
won by $\exists$ if and only if \(v_n\in V_{\forall}\) (i.e.\ if
$\forall$ is stuck); infinite plays are won by $\exists$ if and only
if
\(\max\{p\mid \forall j\in\mathbb{N}.\,\exists k\geq j.\, \Omega(v_k) =
p\}\) is even.  A \emph{(history-free)} $\exists$-\emph{strategy} is a
partial function \(s\colon V_\exists\rightharpoonup V\) that assigns
moves to $\exists$-nodes.  A play \emph{follows} a strategy \(s\) if
for all \(i\geq 0\) such that \(v_i\in V_{\exists}\),
\(v_{i+1}=s(v_i)\)\spnote{@Daniel Müsste heißen, v_i+1 = s(v_i) wenn s(v_i) definiert und ansonsten muss es ein endliches play sein, und v_i = v_n, dann kann s(v_i) undefiniert sein.}.
An $\exists$-strategy \emph{wins} a node
\(v\in V\) if $\exists$ wins all plays that start at \(v\) and follow
\(s\).  

We fix a closed formula $\varphi$ for the remainder of this section
and let $\mathsf{cl}=\mathsf{cl}(\varphi)$ denote its (Fischer-Ladner)
\emph{closure}~\cite{Kozen83}\spnote{Ich denke das ist die gleiche closure, die erneute Erwähnung sollte also nicht kursiv sein und das Zitat an der ersten Erwähnung.}.

\begin{defn}[Model checking games, CGF semantics]
Given a model~$\mathcal{M}$ consisting of 
a concurrent game frame $\mathcal{F}=(W,m,f)$ and a valuation $\rho$,
the \emph{model checking game} 
$\gamecgf_{\mathcal{M},\varphi}=(V,E,\Omega)$ is
the parity game defined by the following table, where
game nodes $v\in V=V_\exists\cup V_\forall$ are of the shape $v=(w,\psi)\in W\times \mathsf{cl}$
or $v=(w,\psi,s_C)\in W\times \mathsf{cl}\times\Pi_C$; in the latter case,
we require $\psi$ to be a modality.
\begin{center}
\begin{tabular}{|c|c|c|c|}
\hline
node & owner & moves to & priority \\
\hline
$(w,\top)$ & $\forall$ & $\emptyset$ & $0$\\
$(w,\bot)$ & $\exists$ & $\emptyset$ & $0$\\
$(w,p)$ & $\exists$ & $\{(w,p)\mid w\in \rho(p)\}$ & $0$\\
$(w,\neg p)$ & $\forall$ & $\{(w,\neg p)\mid w\in \rho(p)\}$ & $1$\\
$(w,\varphi\land\psi)$ & $\forall$ & $\{(w,\varphi),(w,\psi)\}$ & $0$\\
$(w,\varphi\lor\psi)$ & $\exists$ & $\{(w,\varphi),(w,\psi)\}$ & $0$\\
$(w,\eta X.\,\psi)$ & $\exists$ & $\{(w,\psi[\eta X.\,\psi/X])\}$ & $\mathsf{ad}(\eta X.\,\psi)$\\
$(w,\CLbox{C}\psi)$ & $\exists$ & $\{(w,\CLbox{C}\psi,s_C)\mid s_C\in\Pi_C(w)\}$ & $0$\\
$(w,\CLdia{C}\psi)$ & $\forall$ & $\{(w,\CLdia{C}\psi,s_C)\mid s_C\in\Pi_C(w)\}$ & $0$\\
$(w,\CLbox{C}\psi,s_C)$ & $\forall$ & $\{(f(w,s_C,s_{\overline{C}}),\psi)\mid s_{\overline{C}}\in \Pi_{\overline{C}}(w)\}$ & $0$\\
$(w,\CLdia{C}\psi,s_C)$ & $\exists$ & $\{(f(w,s_C,s_{\overline{C}}),\psi)\mid s_{\overline{C}}\in \Pi_{\overline{C}}(w)\}$ & $0$\\

\hline
\end{tabular}
\end{center}
\end{defn}
Thus $\gamecgf_{\mathcal{M},\varphi}$ is a parity game with at most $
|W|\times |\mathsf{cl}|\times (|\Pi|+1)$ nodes and at most $\mathsf{ad}(\varphi)$
priorities. We point out that the model checking game
has winning (resp. losing) self loops at nodes $(w,p)$ (resp. $(w,\neg p)$) such that $w\in\rho({p})$.

\begin{lemma}
Let $\mathcal{F}$ be a concurrent game frame with set~$W$ of states,
and let $w\in W$ be a state. Then we have
$w\in\sem{\varphi}_{\mathcal{F}}$ if and only if the existential player wins the node $(w,\varphi)$ in
$\gamecgf_{\mathcal{M},\varphi}$.
\end{lemma}

\begin{defn}[Model checking games, EF semantics]\upshape
  If the model $\mathcal{M}$ is given as an effectivity frame
  $\mathcal{F}=(W,e)$, the model checking game
  $\gameef_{\mathcal{M},\varphi}=(V=V_\exists\cup V_\forall,E,\Omega)$ is
  defined in the same way as for concurrent game frames, but with
  nodes $v=(w,\psi,U)\in W\times\mathsf{cl}\times \Pow(W)$ replacing
  the nodes $(w,\psi,s)\in W\times \mathsf{cl}\times\Pi$ (with~$\psi$
  of the form $\CLbox{C}\psi'$ or $\CLdia{C}\psi'$). The modal moves
  are given by the table below; all other moves, and also node
  ownership and the priority assignments, are as in the CGF model
  checking game given above.
\begin{center}
\begin{tabular}{|c|c|c|c|}
\hline
node & owner & moves to & priority \\
\hline
$(w,\CLbox{C}\psi)$ & $\exists$ & $\{(w,\CLbox{C}\psi,U)\mid U\in e(w,C)\}$ & $0$\\
$(w,\CLdia{C}\psi)$ & $\forall$ & $\{(w,\CLdia{C}\psi,U)\mid U\in e(w,C)\}$ & $0$\\
$(w,\CLbox{C}\psi,U)$ & $\forall$ & $\{(v,\psi)\mid v\in U\}$ & $0$\\
$(w,\CLdia{C}\psi,U)$ & $\exists$ & $\{(v,\psi)\mid v\in U\}$ & $0$\\
\hline
\end{tabular}
\end{center}
\end{defn}
In this case, $\gameef_{\mathcal{M},\varphi}$ is a parity game with at most $
|W|\times |\mathsf{cl}|\times (2^{|W|}+1)$ nodes and at most $\mathsf{ad}(\varphi)$ priorities.

\begin{corollary}
Let $\mathcal{F}$ be an effectivity frame with set~$W$ of states,  and
let $w\in W$ be a state. 
Then we have $w\in\sem{\varphi}_{\mathcal{F}}$ if and only if the existential player wins the node $(w,\varphi)$ in
$\gameef_{\mathcal{M},\varphi}$.
\end{corollary}
\noindent We point out that the above reductions to parity games make all transitions in models explicit,
leading to games with a relatively large number of game nodes.
Next, we present an alternative solution method that directly evaluates (sub)formulas over
the state space $W$. This method does not explicitly construct the model checking 
games, and thereby avoids the blowup in state space incurred by the reduction to parity games.
\begin{defn}[One-step evaluation]\label{defn:osfun}
Given sets $V\subseteq W$ and $\mathbf{X}=X_0,\ldots,X_k\subseteq V\times\mathsf{cl}$, we define a monotone function 
$\mathsf{prop}_V:\Pow(V\times\mathsf{cl})^{k+1}\to \Pow(V\times\mathsf{cl})$ 
evaluating propositional operators over argument sets $\mathbf{X}=(X_0,\ldots,X_k)$ by 
\begin{align*}
\mathsf{prop}_V(\mathbf{X})=\,&V\times\{\top\}\cup \{(w,p)\mid w\in \rho(p)\}\cup
\{(w,\neg p)\mid w\notin \rho(p)\}\,\cup\\
& \{(w,\varphi\land\psi)\mid \{(w,\varphi),(w,\psi)\}\subseteq X_0\}\,\cup\\
& \{(w,\varphi\lor\psi)\mid \{(w,\varphi),(w,\psi)\}\cap X_0\neq\emptyset\}\,\cup\\
& \{(w,\eta X.\,\psi)\mid (w,\psi[\eta X.\,\psi/X])\in X_{\mathsf{ad}(\eta X.\,\psi)}\}
\end{align*}
Building on this, we define functions
$f^{\mathsf{cgf}}_V:\Pow(V\times\mathsf{cl})^{k+1}\to \Pow(V\times\mathsf{cl})$ 
and 
$f^{\mathsf{ef}}_V:\Pow(V\times\mathsf{cl})^{k+1}\to \Pow(V\times\mathsf{cl})$ 
that evaluate formulas step by step using CGF semantics and EF semantics,
respectively; formally, we put
\begin{align*}
f^{\mathsf{CGF}}_V(\mathbf{X})=&\, \mathsf{prop}_V(\mathbf{X})\cup\{(w,\CLbox{C}\psi)\mid \exists s_C\in\Pi^w_C.\,\forall
s_{\overline{C}}\in \Pi^w_{\overline{C}}.\,(f(w,s_C,s_{\overline{C}}),\psi)\in X_{0}\}\\
& \,\cup\,\{(w,\CLdia{C}\psi)\mid \forall s_C\in\Pi^w_C.\,\exists
s_{\overline{C}}\in \Pi^w_{\overline{C}}.\,(f(w,s_C,s_{\overline{C}}),\psi)\in X_{0}\}\\
f^{\mathsf{EF}}_V(\mathbf{X})
=&\, \mathsf{prop}_V(\mathbf{X})\cup\{(w,\CLbox{C}\psi)\mid \exists U\in e(w,C).\,(U\times\{\psi\})\subseteq X_0\}\,\cup\\
& \,\{(w,\CLdia{C}\psi)\mid \forall U\in e(w,C).\,(U\times\{\psi\})\cap X_0\neq\emptyset\}
\end{align*}

\end{defn}

\begin{defn}\label{def:fp}
Given $V\subseteq W$ and a monotone function $f:\Pow(V\times\mathsf{cl})^{k+1}\to \Pow(V)\times\mathsf{cl}$, we define the nested fixpoints
\begin{align*}
\mathsf{E}_f &= \eta_k X_k.\,\eta_{k-1} X_{k-1}.\,\ldots.\,\nu X_0. f(X_0,\ldots, X_k)\\
\mathsf{A}_f &= \overline{\eta_k} X_k.\,\overline{\eta_{k-1}} X_{k-1}.\,\ldots.\,\mu X_0. \overline{f}(X_0,\ldots,X_k)
\end{align*}
where $\eta_i=\nu$ for even $i$ and $\eta_i=\mu$ for odd $i$; also, $\overline{\nu}=\mu$
and $\overline{\mu}=\nu$. We use $\overline{f}$ to denote the dual function to $f$.
These functions evaluate parity objectives over the functions $f$ and $\overline{f}$,
intuitively solving the associated model checking game by \emph{inlining} intermediate game nodes
that arise from the necessity to evaluate modalities; crucially, the domain of these fixpoint computations is
$\Pow(V)\times\mathsf{cl}$ rather than the full set of nodes from the model checking game.
\end{defn}

\begin{lemma}[\cite{DBLP:conf/concur/HausmannS19}]
  Let $\mathcal{F}$ be a CGF with set~$W$ of states, and let $w\in W$
  be a state. Then $w\in\sem{\varphi}_{\mathcal{F}}$ if and only if
  $(w,\varphi)\in \mathsf{E}_{f^{\mathsf{CGF}}_W}$; if $\mathcal{F}$
  instead is an effectivity frame, then we have
  $w\in\sem{\varphi}_{\mathcal{F}}$ if and only if
  $(w,\varphi)\in \mathsf{E}_{f_W^{\mathsf{EF}}}$.
\end{lemma}
The above implies a model checking algorithm that computes the nested
fixpoints in \cref{def:fp} iteratively. This algorithm can be extended
to support \emph{local} model checking, i.e.\ may avoid exploring the
entire model if satisfaction of the target formula in a target state
can be determined early~\cite{DBLP:conf/concur/HausmannS19}.


\section{Implementation}


\subsection{COOL - The Coalgebraic Ontology Logic Solver}

The model checking implementation is set within the framework provided
by COOL / COOL 2 / COOL-MC - the COalgebraic Ontology Logic solver - a coalgebraic
reasoner and model checker for modal fixpoint logics
\cite{DBLP:conf/cade/GorinPSWW14,COOL2,COOLMC}, implemented in OCaml.  We
contribute implementations of the translation from CGFs to effectivity
frames and instantiations of the generic tool to both CGFs and
effectivity frames.  In the COOL model checking framework, models are
represented as functions \(f\colon W \to FW\) where \(F\) is a set
constructor obtained from the grammar
\[F,G ::= M \mid \mathit{id} \mid \Pow F \mid F\times G \mid M\to G\]
where \(M\) ranges over sets.  The application \(FX\) of a set
constructor to a set \(X\) is defined as follows: \(MX = M\),
\(\mathit{id}\, X = X\), \(\Pow FX = \Pow(FX)\),
\((F\times G)\,X = FX\times GX\), \((M\to G)\,X = M\to GX\), where
\(\Pow(FX)\) is the power set of \(FX\) and \(M\to GX\) is the set of
functions with domain~\(M\) and codomain \(GX\). Functions of type
$W\to FW$ are referred to as
\emph{$F$-coalgebras}~\cite{DBLP:journals/tcs/Rutten00} but we refrain from delving into the general theory.
Effectivity frames are then $F$-coalgebras for
\(F = \Pow\Atoms\times ({\Pow(\Agents)}\to \Pow\Pow)\), while
concurrent game frames are $F$-coalgebras for
\(F = \Pow\Atoms \times \mathbb{N}^\Agents \times (\Pi\to id)\). 

As the datatype to represent all types of models,
COOL defines the algebraic datatype \lstinline|functor_element|.
\mycomment{
	\begin{minted}{ocaml}
		type functor_element =
		(* basic set elements *)
		| Identifier of string
		| INT of int
		| RAT of int * int
		(* elements of constructed sets *)
		| Tuple of functor_element list
		| Set of functor_element list
		| Function of (functor_element, functor_element)Hashtbl.t
	\end{minted}
}
\mycomment{
A value of this type represents a model, e.g. an effectivity frame \(f\colon W \to \Pow(\Atoms)\times(\Pow(\Agents)\to\Pow\Pow(W))\) is represented by a
value of type

\begin{minipage}[t]{\textwidth}
\begin{minted}{ocaml}
  Function of ( Id of string, (* state from W *), 
    Tuple of (
      Set of (Id of string), (* Atoms, valid at Inputstate *)
      Function of (
        Set of (INT of int) list, (* P(Ag) coalition *)
        Set of (Set of (Id of string) list) list (* P(P(W)) *)
      )Hashtbl.t
    )
  )Hashtbl.t
\end{minted}
\end{minipage}
}

\subsection{Parity Game Model Checking}

In COOL-MC, the generation of parity games is implemented as a higher order function
that traverses the input model and formula and translates all
connectives into game nodes as described in \cref{sec:mc},
implementing the interpretation of modal operators using a function it
receives as an argument. Each occurrence of a coalitional modality
generates a region consisting of inner nodes; the
remaining nodes are called outer nodes, and form a further
region.  The structure of inner nodes differs depending on the
semantics. Each region has an initial node (either the
occurrence of a modality or the root node of the game). The type for
game nodes consists of one algebraic type encoding the structure
of inner nodes, and a further algebraic type encoding whether a node
is an inner node or not.

The parity game thus created then can be solved using any parity game
solver; COOL-MC, which we employ for the experiments in this work, uses the implementation of
Zielonka's algorithm that is provided by PGSolver~\cite{pgsolvegit}.

\subsection{Local Model Checking}\label{sec:local-mc}


Within the COOL-MC framework, we also implement the
alternative local model checking algorithm described in \cref{sec:mc}, again
realized as 
a higher order function that receives the semantic function for
modalities as an argument,
as in Definition~\ref{defn:osfun}, and then computes
the relevant fixpoints by Kleene approximation; indeed, both coalitional modalities
$\CLbox{C}$ and $\CLdia{C}$ are implemented within the same
function. Recall that the semantics of \(\CLbox{C}\psi\) is given by
\begin{align*}
\exists s_C\in\Pi_C(w).\,\forall s_{\overline{C}}\in
\Pi_{\overline{C}}(w).\,(f(w,s_C,s_{\overline{C}}),\psi)\in X_{0}
\end{align*}
in CGF semantics (using notation introduced in \cref{sec:amc}).  
The quantifiers in this formula are implemented by traversing all
joint moves of~$C$ and~$\overline C$, respectively.

The implementation of the semantics of coalitional modalities for
local model checking over CGFs (see \cref{lst:cgf_local}) enumerates
all possible moves using the recursive function \lstinline|try_out|
until it either has found a witnessing move for the coalition, or it
has proved that there is not such a move.  This is achieved by
advancing the move of the coalition if the reached world with the
\lstinline|curr_move| does not satisfy the argument formula
(\lstinline|x| not contained in \lstinline|argset|).  This correspond
the existential search for the move of the coalition in the concurrent
game structure semantics definition.  When the current move of the
coalition reaches a satisfying world then instead the move of the
opposition \lstinline|d_bar| is advanced using \lstinline|next_move|.
The recursion stops if the individual moves have hit the
\lstinline|move_bounds| so that all relevant moves have been explored.

\begin{listing}[htbp]
	\begin{minted}{ocaml}
let coalition_modal_pred _ form xi_of_x argset argset_dual =
  let open Output in
  let* (box, d) =
    match form.HCFml.node with
    | F.HCENFORCES (d, _) -> return (true, d)
    | F.HCALLOWS (d, _) -> return (false, d)
    | _ -> Error (InvalidInput ( "coalition_modal_pred"))
  in
  let argset = if box then argset else argset_dual in
  begin (* xi_of_x = (moves_1, ..., moves_n, result_fun) *)
    let* tuple = M.fe_tuple_to_list xi_of_x in
    let n = ((length tuple) - 1) in (* number of agents *)
    ...
    let rec try_out curr_move =
      let x = Hashtbl.find result_fun curr_move in
      if mem x argset then
        begin
          let next = next_move_for move_bounds d_bar curr_move in
          if next = curr_move then return true else try_out next
        end
      else
        begin
          let next = next_move_for move_bounds d curr_move in
          if next = curr_move
          then return false
          else try_out (reset d_bar next)
        end
    in
    (* start exploring with move where everyone chooses 1 *)
    try_out (List.init n (fun _ -> 1) ) 
  end ||> fun x -> if box then x else not x
\end{minted}
\vspace{-1em}
	\caption{Semantics of the coalitional modalities for local model checking in concurrent game structure semantics}\label{lst:cgf_local}
\end{listing}

For comparison, in EF semantics, the semantics of $\CLbox{C}\psi$ is
realized by implementing the function
\begin{align*}
\exists U\in e(w,C).\,(U\times\{\psi\})\subseteq X_0.
\end{align*}
\cref{lst:ef_local} shows the new implementation in COOL.  Due to the
generic representation of the models and formulas, a few lines
of code are needed to extract the \lstinline|coalition| and
\lstinline|eff_func| from the model and formula representation.  Then
the \lstinline|capabilities| of the coalition corresponding to the
\(e(w,C)\) in the semantics definition are looked up
(\lstinline|H.find|) in the hash table underlying the effectivity
function.  The subsequent \lstinline|exists| search then corresponds
to picking the set \(U\) in the semantics, and checking containment
of $U$ in \lstinline|argset| corresponds to the clause
\(U \subseteq \sem{\varphi}_\sigma\) in the semantics definition.

\begin{listing}[htbp]
\begin{minted}{ocaml}
let coalition_modal_effectivity_pred _ form xi_of_x argset argset_dual =
  let open Output in
  let* (box, coalition) =
    match form.HCFml.node with
    | F.HCENFORCES (d, _) -> return (true, d)
    | F.HCALLOWS (d, _) -> return (false, d)
    | _ -> Error (InvalidInput ( "coalition_modal_effectivity_pred"))
  in
  let argset = if box then argset else argset_dual in
  begin
    let* eff_func = M.fe_function_to_htbl __LINE__ xi_of_x in
    let  capabilities = H.find eff_func (Set (map M.int coalition)) in
    ...
    return (exists (fun succs -> subset succs argset) capabilities)
  end ||> fun x -> if box then x else not x
\end{minted}
\vspace{-1em}
	\caption{Semantics of the coalitional modalities for local model checking in effectivity frame semantics}\label{lst:ef_local}
\end{listing}

\subsection{Converting CGFs to EFs}

To calculate the induced effectivity function at a given state in a
CGF, we implement the translation procedure described in
\cref{def:translation}.  As the set of states stays the same, the
overall conversion is just a matter of converting the structure at
each state individually.  We define a function
\lstinline|effectivity_function_of_game_form|, which we iterate over all states to
compute the resulting effectivity function (see
\cref{lst:conversion}).
\begin{listing}[htbp]
\begin{minted}{ocaml}
let effectivity_function_of_game_form game_form =
  match game_form with
  | Tuple (* (moves_1, ..., moves_n, result_func)*) ->
     ...
     let effectivity = H.create (number_of_agents) in
     ...
     let move_is_extension grand_move coal coal_move = 
       for_all2 (fun x y -> (nth grand_move (x-1)) = y) coal coal_move
     in
     let result_states coal coal_move = (* reached with partial move *)
       let result_state_for_extension grand_move =
         if move_is_extension grand_move coal coal_move
         then Some (result_func grand_move) else None
       in
       Set (filter_map result_state_for_extension all_moves)
     in
     let moves_for_coalition coal =
       let actions_for_coal = map (fun x ->(*{1,...,moves_x}*)) coal in
       all_combinations actions_for_coal
     in
     let enforced_sets_for_coalition coal =
       Set (map (result_states coal) (moves_for_coalition coal))
     in
     iter (fun coal ->
         H.add effectivity coal (enforced_sets_for_coalition coal)
       ) all_coalitions; Function effectivity
   \end{minted}
   \vspace{-1em}
   \caption{Conversion from CGFs to effectivity frames}\label{lst:conversion}
\end{listing}
As the effectivity function assigns an effectivity set to every coalition, the implementation creates an empty hash table and then iterates over the set \lstinline|all_coalitions| and calls the internal function \lstinline|enforced_sets_for_coalition| to calculate the effectivity.
This is done by calculating all admissible joint moves of the given coalition.
Then the set \lstinline|all_moves| is read off from the domain of \lstinline|result_func| and filtered for each of the partial moves of the coalition to get all grand moves extending the partial move of the coalition.
Mapping the \lstinline|result_func| over the resulting set gives the set of  states reached by the joint move.
After the effectivity functions have been computed, our conversion procedure filters the effectivity functions to only contain minimal sets.
This operation ensures minimal runtime and space usage in all further processing of the resulting effectivity frames.

\section{Experiments}

We conduct various experiments, evaluating the performance of the different 
implementations of model checking for the alternating-time $\mu$-calculus
in comparison with each other, and (on the ATL
fragment) in comparison to the state-of-the-art ATL model checker MCMAS.
In more detail, we run benchmarks on the following implementations.
\begin{itemize}
\item 
Our implementation of concurrent game frame semantics model checking as an instance of COOL. We denote these algorithms by
CGF$_g$ (model checking by parity game reduction, using the implementation of
Zielonka's algorithm provided by PGSolver~\cite{pgsolvegit} to solve the resulting games) and
CGF$_l$ (local model checking), respectively.
\item Our implementation of model checking for effectivity frame semantics as an instance of COOL, writing
EF$_g$ (model checking by parity game reduction) and
EF$_l$ (local model checking), respectively, to denote the pure model checking algorithms;
the combined algorithms that first perform the transformation from CGF to EF and subsequently model check are denoted by EFC$_g$ and EFC$_l$.
\item The symbolic multi-agent system model checker MCMAS (version 1.3.0)~\cite{DBLP:journals/sttt/LomuscioQR17}; in contrast to the implementations within COOL, MCMAS handles additional epistemic modalities, and uses ordered binary decision diagram~(OBDD) representations of models and formulas. On the other hand, MCMAS does not support the full AMC, so that the comparison between our implementations in COOL-MC and MCMAS is
restricted to ATL.
\end{itemize}
\noindent In benchmarks, we use
hyperfine~\cite{Peter_hyperfine_2023} to
average the measured values over at least five executions and set a
timeout of 200 seconds.  All experiments have been executed on a
machine with an AMD Ryzen 7 2700 CPU and 32GB of RAM.  
An artifact containing the source code of our implementations, 
evaluation scripts and benchmarking sets for all experiments 
described below will be made available.

\subsection{Random Formulas over Random Games}
In a first experiment, we evaluate our implementations within COOL on
random AMC formulas of increasing size on random concurrent game
frames, each with~10 states, 4 agents, and either~2 or~10 moves per
agent (due to the differences in model representation, a meaningful
comparison with MCMAS on random models does not appear feasible). 

\subsection{Castle Game}
The \emph{castle game} has been used for benchmarking in previous work
on ATL model
checking~\cite{DBLP:journals/logcom/PileckiBJ17,unboundedatl2021}.
The game is parametrized over the number of castles and the health
points all castles start with.  Each castle has a corresponding knight
that can, in each turn, either be sent out to attack another castle or
stay and defend the castle.  In each turn, all knights decide
concurrently which other castle they want to attack or if they want to
stay at their castle and defend.  A knight who has attacked in one
turn needs to stay and rest in the next turn.  A castle that has its
knight defending it or resting can block one attack.  Each unblocked
attack on a castle reduces that castle's number of health points by
one. When no health points are left, the castle has lost the game and
can no longer attack; this situation is indicated by propositional
atoms $\mathsf{lost}_{a}$, where $a$ is a knight.

For the castle game we check the following AMC formulas 
(which are expressible in ATL) for satisfaction in the initial state. The formula $\nu X.\, \neg\mathsf{lost}_\agent \wedge \CLbox{\{\agent\}}X$
expresses that the knight \(\agent\) has a strategy ensuring that her castle never gets destroyed.
We check this formula for each \(\agent \in \Agents\). Moreover, the formula
\[\textstyle\mu X.\, ((\bigwedge_{\agent \in C}\neg\mathsf{lost}_\agent) \wedge (\bigwedge_{\agent \in \Agents\setminus C}\mathsf{lost}_\agent)) \vee \CLbox{C}X\]
expresses that the coalition \(C\) has a joint strategy to ensure that
all other castles are eventually destroyed while none of the allied
castles (belonging to $C$) are destroyed.  We check this formula for
one coalition of each size.

The castle game has the property that almost none of the joint moves
are equivalent, i.e.\ almost all joint moves lead to a different
outcome.  The specification in MCMAS
uses separated local states of the agents.
\mycomment{
\begin{listing}[htb]
  \begin{lstlisting}
Agent ag1
  Lobsvars = {  };
  Vars: ready : boolean; hp : 0 .. 2; end Vars 
  Actions = { defend, rest, dead, attack_2 }; 
  Protocol:
    ((hp) >= (1) and ready = true) : { defend, attack_2 };
    ((hp) >= (1) and ready = false) : { rest };
    Other : { dead };  end Protocol 
  Evolution:
    hp = (hp) - (1) and ready = false if (Action = attack_2 and
      ((ag2.Action = attack_1 and (hp) >= (1)) or (hp = 1 and
      ag2.Action = attack_1)));
    hp = (hp) - (0) and ready = true if (!(Action = attack_2) and
      ((ag2.Action = attack_1 and (hp) >= (0)) or (hp = 0 and
      ag2.Action = attack_1)));
    ready = false if ((hp) >= (1) and (!(ag2.Action = attack_1) and
      Action = attack_2));
    ready = true if ((hp) >= (1) and (!(ag2.Action = attack_1) and
      !(Action = attack_2))); end Evolution
end Agent
  \end{lstlisting}
\vspace{-20pt}
  \caption{MCMAS encoding of an agent in the two-castle two-health-point game}\label{lst:castles-agent}
\end{listing}
}
The encoding of the castle game with \(n\) castles and \(h\) health points
in COOL uses \((2 \times H)^n\) as state space where \(2 = \{t, f\}\) and \(H = \{x \mid 0 \le x \le h\}\).

\subsection{Modulo Game}
The second game we consider is constructed to showcase the benefit of effectivity frames in cases with  many equivalent joint moves. The \emph{modulo game} is parameterized over the number of agents, the number~$n$ of moves per agent, and a base number for the modulo calculation, which is also the number of states.
Each state is uniquely identified by satisfaction of an atom \(p_i\) that witnesses that the current state corresponds to a sum of \(i\).
  In each turn, the agents concurrently choose a number between one and~$n$.
  The game then moves to the state corresponding to the sum of the played numbers plus the number of the previous state, modulo the base.

For the modulo game, we check the following formulas for satisfaction in the initial state (in which the mentioned sum starts at \(0\)). The formula
\[\varphi_1:=\textstyle\bigwedge_{0 \le i < \mathsf{base}}\mu X.\, p_i \vee \CLbox{C}X\]
expresses that the coalition \(C\) has a joint strategy to reach any state eventually;
this property can be expressed in ATL. On the other hand, the formula
\[\varphi_2:=\nu X.\, \mu Y.\,(X \wedge (p_0 \vee \CLbox{C}Y) \wedge (p_{\mathsf{base} / 2} \vee \CLbox{C}Y))\]
expresses the Büchi property that coalition \(C\) has a strategy to ensure that a sum of \(0\) as well as a sum of \(\mathsf{base} / 2\) occur infinitely often.
We check both formulas for one coalition of each size.
We use the first formula for comparison with MCMAS.
The second formula requires the full AMC and is hence not covered by MCMAS, so it is used only for benchmarking the different COOL implementations. 
We fix 10 as the base throughout.

The (explicit-state) encoding of the modulo game in COOL is
straight-forward, having ten states and the expected transition
function and valuation.

\subsection{Results}

The results of the experiment on \emph{random formulas} are summarized in~\cref{fig:random-bench}.
Superscripts indicate the number of actions available to each of the agents; all models generated have 10 states.
The vertical axis denotes the mean cumulative runtime including model conversion time  for one run, i.e. an average over 25 instances of the model checking problem, each instance  averaged over at least 5 runs. The horizontal axis denotes the size of formulas in terms of the number of connectives.

As indicated in the theoretical discussion of the individual model checking algorithms, inference using EF starts with a runtime offset owed to the conversion from CGF models to EF model.
At small problem sizes, this overhead makes EF-based model checking slower than direct model checking on CGFs; however, the runtime eventually breaks even, with the EF-based algorithm  outscaling the CGF implementation.
In our benchmark parameters, this happens at a formula size near 16 for the local variants,
and for much larger formulas for the game based implementations.

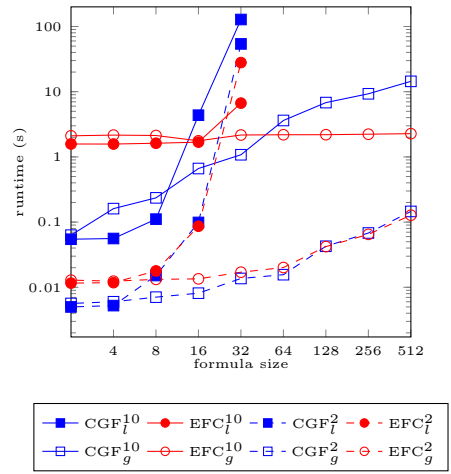
\begin{wrapfigure}{r}{0.5\textwidth}
      \centering
      \begin{tikzpicture}
        \begin{loglogaxis}[
          mark options={solid},
          minor tick num=1,
          xtick distance = 2,
          every axis y label/.style=
          {at={(ticklabel cs:0.5)},rotate=90,anchor=center},
          every axis x label/.style=
          {at={(ticklabel cs:0.5)},anchor=center},
          tiny,
          width=\linewidth,
          height=17em,
          transpose legend,
          legend columns=2,
          legend style={at={(0.5,-0.2)},anchor=north},
          log basis x=2,
          log basis y=10,
          log ticks with fixed point,
          xlabel={formula size},
          ylabel={runtime (s)},
          xmin=2,
          xmax=512,
          ymin=0,
          ymax=200,
          legend entries={CGF$_l^{10}$, CGF$_g^{10}$, EFC$_l^{10}$, EFC$_g^{10}$, CGF$_l^{2}$, CGF$_g^{2}$, EFC$_l^{2}$, EFC$_g^{2}$}]
  
          \addplot[blue, mark=square*] table [col sep=comma, x=parameter, y=mean] {benchmarks/modelsize10states10actionsmergedlocalCGF.csv};
          \addplot[blue, mark=square] table [col sep=comma, x=parameter, y=mean] {benchmarks/modelsize10states10actionsmergedsmallCGF.csv};
          \addplot[red, mark=*] table [col sep=comma, x=parameter, y=mean] {benchmarks/modelsize10states10actionsmergedlocalEF.csv};
          \addplot[red, mark=o] table [col sep=comma, x=parameter, y=mean] {benchmarks/modelsize10states10actionsmergedsmallEF.csv};
          \addplot[blue, mark=square*, dashed] table [col sep=comma, x=parameter, y=mean] {benchmarks/modelsize10states2actionsmergedlocalCGF.csv};
          \addplot[blue, mark=square, dashed] table [col sep=comma, x=parameter, y=mean] {benchmarks/modelsize10states2actionsmergedsmallCGF.csv};
          \addplot[red, mark=*, dashed] table [col sep=comma, x=parameter, y=mean] {benchmarks/modelsize10states2actionsmergedlocalEF.csv};
          \addplot[red, mark=o, dashed] table [col sep=comma, x=parameter, y=mean] {benchmarks/modelsize10states2actionsmergedsmallEF.csv};
  
        \end{loglogaxis}
      \end{tikzpicture}
  \caption{Model checking random formulas} \label{fig:random-bench}
    \end{wrapfigure}

The \emph{conversion times} for the modulo game are shown
in~\cref{fig:conv-times-modulo}. The effort required for converting from a 
concurrent game frame to an effectivity frame grows  with both the
 number of states in the model and  the number of moves.
\mycomment{Increasing the health points in the castles game 
(\cref{fig:conv-times-castle}) has no effect on the size of the transition function.
However, these models have roughly \(\mathsf{max\_hp}^{|\Agents|}\) many 
states so an increase in starting health leads to a large increase in the 
number of states which in turn all have the large transition function attached.}
The modulo game always has exactly 10 states, and increasing the 
number of moves per agent leads to a blowup in the individual transition functions.
The transition function has an entry for every joint move, so its domain has 
size \(\mathsf{moves\_per\_agent}^{|\Agents|}\), leading to the drastic growth 
visible in \cref{fig:conv-times-modulo}.
The figure also shows the expected increase in conversion time when the number
of agents is increased. The behaviour for the castles game is similar but scaled up
due to the larger state space.

\begin{figure}[htbp]\mycomment{
  \begin{minipage}[t]{.5\linewidth}
    \centering
    \begin{tikzpicture}
      \begin{semilogyaxis}[
        mark options={solid},
        minor tick num=1,
        xtick distance = 2,
        every axis y label/.style=
        {at={(ticklabel cs:0.5)},rotate=90,anchor=center},
        every axis x label/.style=
        {at={(ticklabel cs:0.5)},anchor=center},
        tiny,
        width=\linewidth,
        height=17em,
        transpose legend,
        legend columns=2,
        legend style={at={(0.5,-0.2)},anchor=north},
        ymode=log,
        log ticks with fixed point,
        xlabel={starting health points},
        ylabel={runtime (s)},
        xmin=1,
        xmax=10,
        ymin=0,
        ymax=100,
        legend entries={$2$ castles, $3$ castles, $4$ castles, $5$ castles}]

        \addplot table [col sep=comma, x=parameter_hp, y=mean] {benchmarks/castles-conversions-2.csv};
        \addplot table [col sep=comma, x=parameter_hp, y=mean] {benchmarks/castles-conversions-3.csv};
        \addplot table [col sep=comma, x=parameter_hp, y=mean] {benchmarks/castles-conversions-4.csv};
        \addplot table [col sep=comma, x=parameter_hp, y=mean] {benchmarks/castles-conversions-5.csv};
      \end{semilogyaxis}
    \end{tikzpicture}
    \caption{Castle game conversion time}\label{fig:conv-times-castle}
  \end{minipage}%
  \,\,\,}
  \begin{minipage}[t]{.5\linewidth}
    \centering
    \begin{tikzpicture}
      \begin{semilogyaxis}[
        mark options={solid},
        minor tick num=1,
        xtick distance = 2,
        every axis y label/.style=
        {at={(ticklabel cs:0.5)},rotate=90,anchor=center},
        every axis x label/.style=
        {at={(ticklabel cs:0.5)},anchor=center},
        tiny,
        width=\linewidth,
        height=17em,
        transpose legend,
        legend columns=2,
        legend style={at={(0.5,-0.2)},anchor=north},
        ymode=log,
        log ticks with fixed point,
        xlabel={number of moves per agent},
        ylabel={runtime (s)},
        xmin=2,
        xmax=10,
        ymin=0,
        ymax=100,
        legend entries={$2$ agents, $3$ agents, $4$ agents, $5$ agents}]

        \addplot table [col sep=comma, x=parameter_moves, y=mean] {benchmarks/modulos-conversions-2.csv};
        \addplot table [col sep=comma, x=parameter_moves, y=mean] {benchmarks/modulos-conversions-3.csv};
        \addplot table [col sep=comma, x=parameter_moves, y=mean] {benchmarks/modulos-conversions-4.csv};
        \addplot table [col sep=comma, x=parameter_moves, y=mean] {benchmarks/modulos-conversions-5.csv};
      \end{semilogyaxis}
    \end{tikzpicture}
    \caption{Modulo game conversion time}\label{fig:conv-times-modulo}
  \end{minipage}
  \begin{minipage}[t]{.5\linewidth}
    \centering
    \begin{tikzpicture}
      \begin{semilogyaxis}[
        mark options={solid},
        minor tick num=1,
        xtick distance = 2,
        every axis y label/.style=
        {at={(ticklabel cs:0.5)},rotate=90,anchor=center},
        every axis x label/.style=
        {at={(ticklabel cs:0.5)},anchor=center},
        tiny,
        width=\linewidth,
        height=17em,
        transpose legend,
        legend columns=2,
        legend style={at={(0.5,-0.2)},anchor=north},
        ymode=log,
        log ticks with fixed point,
        xlabel={starting health points},
        ylabel={runtime (s)},
        xmin=1,
        xmax=10,
        ymin=0,
        ymax=200,
        legend entries={MCMAS, CGF$_l$, CGF$_g$, EF$_l$, EF$_g$, EFC$_g$}]

        \addplot[mark=diamond*] table [col sep=comma, x=parameter_hp, y=mean] {benchmarks/castlesMCMAS-4.csv};
        \addplot[mark=square*, blue] table [col sep=comma, x=parameter_hp, y=mean] {benchmarks/castlesCOOL-4-CL-local.csv};
        \addplot[mark=square, blue] table [col sep=comma, x=parameter_hp, y=mean] {benchmarks/castlesCOOL-4-CL-small.csv};
        \addplot[mark=*, red] table [col sep=comma, x=parameter_hp, y=mean] {benchmarks/castlesCOOL-4-CLEF-local.csv};
        \addplot[dashed, mark=o, red] table [col sep=comma, x=parameter_hp, y=mean] {benchmarks/castlesCOOL-4-CLEF-small.csv};
        \addplot[mark=oplus, red] table [col sep=comma, x=parameter_hp, y=mean] {benchmarks/total-castlesCOOL-4-CLEF-small.csv};
      \end{semilogyaxis}
    \end{tikzpicture}
    \caption{Castle game runtime (4 castles)}\label{fig:bench-castles-4}
  \end{minipage}
\vspace{-10pt}

\end{figure}
The average runtime results for evaluating all
described formulas over the \emph{castles game} models are depicted in
\cref{fig:bench-castles-4}; they appear to show
that MCMAS consistently outperforms our implementations in COOL.  Both
tools have similar performance characteristics with respect to
increasing the starting health points affecting the number of states
in the game and also the number of turns necessary to defeat a castle.
The local variants of the COOL model checking perform significantly
worse than the game-based variants on both the CGF and the EF
semantics, possibly due to better scaling of Zielonka's algorithm as
provided by PGSolver in comparison with the naive fixpoint
approximation implemented in COOL.  It is also notable here that while
the model checking on effectivity frames generally performs better --
as seen best in \cref{fig:bench-castles-4} -- the cost of conversion 
outweighs this speedup at this
model size.  If a large number of formulas is to be evaluated on the
same model, the cost of conversion is expected to amortize due to the
slight speedup and might make this approach faster in the end.  One
reason why effectivity frames do not result in more speedup in the castle game 
is presumably that almost no moves are equivalent, so the
resulting EFs are comparatively large.  
\begin{figure}[ht]
   \begin{minipage}[t]{.48\linewidth}
    \centering
    \begin{tikzpicture}
      \begin{semilogyaxis}[
        mark options={solid},
        minor tick num=1,
        xtick distance = 2,
        every axis y label/.style=
        {at={(ticklabel cs:0.5)},rotate=90,anchor=center},
        every axis x label/.style=
        {at={(ticklabel cs:0.5)},anchor=center},
        tiny,
        width=\linewidth,
        height=17em,
        transpose legend,
        legend columns=2,
        legend style={at={(0.5,-0.2)},anchor=north},
        ymode=log,
        log ticks with fixed point,
        xlabel={number of moves},
        ylabel={runtime (s)},
        xmin=2,
        xmax=10,
        ymin=0,
        ymax=100,
        legend entries={MCMAS, CGF$_l$, CGF$_g$, EF$_l$, EF$_g$, EFC$_g$}]

        \addplot[mark=diamond*] table [col sep=comma, x=parameter_moves, y=mean] {benchmarks/modulosMCMAS-4.csv};
        \addplot[mark=square*, blue] table [col sep=comma, x=parameter_moves, y=mean] {benchmarks/modulosCOOL-4-CL-local.csv};
        \addplot[mark=square, blue] table [col sep=comma, x=parameter_moves, y=mean] {benchmarks/modulosCOOL-4-CL-small.csv};
        \addplot[mark=*, red] table [col sep=comma, x=parameter_moves, y=mean] {benchmarks/modulosCOOL-4-CLEF-local.csv};
        \addplot[dashed, mark=o, red] table [col sep=comma, x=parameter_moves, y=mean] {benchmarks/modulosCOOL-4-CLEF-small.csv};
        \addplot[mark=oplus, red] table [col sep=comma, x=parameter_moves, y=mean] {benchmarks/total-modulosCOOL-4-CLEF-small.csv};
      \end{semilogyaxis}
    \end{tikzpicture}
    \caption{Modulo game ($\varphi_1$, 4 agents)}\label{fig:bench-modulo-4}
  \end{minipage}
  \;
  \begin{minipage}[t]{.49\linewidth}
    \centering
    \begin{tikzpicture}
      \begin{semilogyaxis}[
        mark options={solid},
        minor tick num=1,
        xtick distance = 2,
        every axis y label/.style=
        {at={(ticklabel cs:0.5)},rotate=90,anchor=center},
        every axis x label/.style=
        {at={(ticklabel cs:0.5)},anchor=center},
        tiny,
        width=\linewidth,
        height=17em,
        transpose legend,
        legend columns=2,
        legend style={at={(0.5,-0.2)},anchor=north},
        ymode=log,
        log ticks with fixed point,
        xlabel={number of moves},
        ylabel={runtime (s)},
        xmin=2,
        xmax=10,
        ymin=0,
        ymax=100,
        legend entries={CGF$_l$, CGF$_g$, EF$_l$, EF$_g$, EFC$_g$}]

        \addplot[mark=square*, blue] table [col sep=comma, x=parameter_moves, y=mean] {benchmarks/modulosCOOL-4-CL-local-mu.csv};
        \addplot[mark=square, blue] table [col sep=comma, x=parameter_moves, y=mean] {benchmarks/modulosCOOL-4-CL-small-mu.csv};
        \addplot[mark=*, red] table [col sep=comma, x=parameter_moves, y=mean] {benchmarks/modulosCOOL-4-CLEF-local-mu.csv};
        \addplot[dashed, mark=o, red] table [col sep=comma, x=parameter_moves, y=mean] {benchmarks/modulosCOOL-4-CLEF-small-mu.csv};
        \addplot[mark=oplus, red] table [col sep=comma, x=parameter_moves, y=mean] {benchmarks/total-modulosCOOL-4-CLEF-small-mu.csv};
      \end{semilogyaxis}
    \end{tikzpicture}
    \caption{Modulo game ($\varphi_2$, 4 agents)}\label{fig:bench-modulo-mu-4}
  \end{minipage}
\end{figure}
On the other hand,  EF-based model checking performs very well when model checking
 the \emph{modulo games} against the formula 
$\varphi_1$ (\cref{fig:bench-modulo-4}).
The runtime of the EF-based semantics is essentially constant, while the runtime 
under CGF semantics increases with the number of moves per agent.
The constant runtime under EF semantics is as expected, since any coalition has a constant number of non-equivalent joint moves.
The effectivity function hence scales mainly with the number of agents and only 
up to a fixed limit w.r.t.\ the number of moves per agent or the size of the 
coalition.
The two EF implementations in COOL even outperform 
MCMAS with increasing number of moves and agents.
Note however, that the conversion time increases sharply with increasing model 
size (similar to~\cref{fig:conv-times-modulo}), nivellating to some extent the 
performance gain when the system is not initially specified as effectivity frame.
 Spikes in the plots appear
to be due to symmetry in the model structure.
In \cref{fig:bench-modulo-mu-4}, the four COOL implementations are compared on the modulo game with the Büchi property benchmark formula~$\varphi_2$.
As expected, the effectivity-frame-based implementations have roughly constant performance while the concurrent game frame implementations slow down drastically with increased number of moves per agent.

The three benchmarking sets we evaluate all vary different
parameters.  Generally, effectivity frame semantics
plays out its
advantages when there are many equivalent joint moves of a coalition,
as showcased by the constant performance in the modulo game
(e.g.~\cref{fig:bench-modulo-4}).  Compared to model checking random formulas
on concurrent game frames directly, there is a speedup for large input formulas as
seen in~\cref{fig:random-bench}. The BDD-based approach of
MCMAS outperforms COOL on examples with few equivalent joint moves
(see~\cref{fig:bench-castles-4}).  On smaller problems, the
speedup of the EF semantics is diminished by the initial cost of
conversion when the problem is not specified as effectivity frame
directly (see~\cref{fig:conv-times-modulo}).  We note however that
COOL is the only model checker able to handle the full AMC where the
speedup is potentially multiplied due to unfolding of the fixpoints
increasing the effective formula size (see~\cref{fig:bench-modulo-mu-4}).

\section{Conclusions and Future Work}

\noindent We have analysed efficiency gains in model checking for the
alternating-time $\mu$-calculus afforded by converting concurrent game
frames into effectivity frames in a preprocessing step. We have
 evaluated this method in comparison both with a direct
implementation of model checking on concurrent game structures within
the same overall framework and with the state-of-the-art ATL
symbolic model checker MCMAS. Results show favourable performance of
the preprocessing method on large systems that do not have succinct
symbolic representations.

A large part of the computational cost of our method lies in the
preprocessing, so future development will focus on
optimizations of the conversion. In particular, since the conversion
is entirely per-state, there is potential for
parallelization. Moreover, further performance gains might be achievable by converting from concurrent game
frames to effectivity functions by-need, that is, only when a modality
is actually being evaluated on the state at hand.




%
%
%
\bibliographystyle{splncs04}
\bibliography{refs}
\newpage
\end{document}